\documentclass{amsart}
\usepackage{epsfig}
\usepackage{nicefrac}

\newtheorem{theorem}{Theorem}

\usepackage{geometry} 
\usepackage{amsmath}
\usepackage{graphicx}
\usepackage{epsfig}
\usepackage{latexsym}
\usepackage{epsfig}
\newcommand{\PP}{{\mathbb P}}

\newcommand{\cC}{\mathcal C}

\begin{document}

\title[Capturing a tree with few characters]{Capturing a phylogenetic tree when the number of character states varies with the number of leaves}
\author{Mike Steel}
\address{Biomathematics Research Centre, University of Canterbury, Christchurch, NZ}
\email{mike.steel@canterbury.ac.nz}

\begin{abstract}
We show that for any two values $\alpha, \beta >0 $ for which $\alpha+\beta>1$ then there is a value $N$ so that for all $n \geq N$ the following holds. For any binary phylogenetic tree $T$ on $n$ leaves there is a set of $\lfloor n^\alpha \rfloor$ characters that capture $T$, and  for which each character takes at most $\lfloor n^\beta \rfloor$ distinct states. Here `capture' means that $T$ is the unique  perfect phylogeny 
for these  characters.  Our short proof of this  combinatorial result is based on the probabilistic method.
\end{abstract} 

\maketitle

Given a function $f: X\rightarrow S$ let $\pi(f)$ denote the partition of $X$ induced by the equivalence relation $x \sim x'$ if
and only if $f(x)=f(x')$.  If $|\pi(f)| \leq r$ we say that $f$ {\em takes at most} $r$ states (this is equivalent to saying $|f(X)| \leq r$, and such characters are also referred to as `$r$-state characters' elsewhere).
Given a  (unrooted) phylogenetic $X$--tree $T$ (i.e. a tree leaf set $X$ and no vertices of degree 2) $f: X \rightarrow S$ is said to be 
a {\em character} on $X$ and $f$ is {\em convex} on $T$ if the minimal subtrees of $T$ connecting the leaves of each block of the partition $\pi(f)$
are vertex disjoint. The condition of $f$ being convex has a natural interpretation in biology of the character $f$ being
`homoplasy-free' (for details, see \cite{sem}). Now suppose we are given a set $\cC$ of characters on $X$. In this case $T$ is said to be a
{\em perfect phylogeny} for $\cC$ if each of the characters in $\cC$ are convex on $T$.  Moreover, $\cC$ is said to {\em capture} $T$ if
$T$ is the only perfect phylogeny for $\cC$, in which case every non-leaf vertex of $T$ must have degree 3.

Suppose that $\cC$ is a set of $k$ characters, each of which takes at most $r$ states. Then a fundamental inequality states that $k$ must be at least $\lceil (n-3)/(r-1)\rceil $ (Proposition 4.2 of \cite{sem}).  Remarkably,  this lower bound was recently shown  \cite{bor15} to be sharp for every fixed value of $r>1$,  provided that $n \geq N_r$, where $N_r$ is some (increasing)  function of  $r$ (e.g. $N_2=3, N_3=13$ \cite{bor15}).  In other words, for every $r>1$, and every unrooted binary phylogenetic $X$--tree $T$, where $n=|X|\geq N_r$,
  there is a set $\cC$ of $\lceil (n-3)/(r-1)\rceil $ characters  that captures $T$ and with each character in $\cC$ taking at most $r$ states.

In this note, we consider how small $k$ can be  when $r$ is allowed to depend on $n$ (we write $r=r_n$).   From \cite{4char1, 4char2} it is  known that there exists a set $\cC$ of $k=4$ characters for which the associated number $n_r$ of states satisfies  $n/r_n= O(1)$. Thus we focus on the setting where both $r_n$ and $n/r_n$ grow with increasing $n$. More precisely, suppose that we want a set $\cC_n$ consisting of $k_n = \lfloor n^\alpha \rfloor$ characters on $X$, each taking at most $r_n = \lfloor n^\beta \rfloor$ states, to capture some phylogenetic $X$--tree, where $\alpha, \beta >0$.  
   Notice that the inequality
 $k \geq \lceil (n-3)/(r-1)\rceil$ implies that $k_n$ must exceed $n^{1-\beta}$ 
for $n$ sufficiently large, thus $\alpha + \beta >1$.  We show here that any value of $\alpha, \beta >0$ with $\alpha+ \beta>1$ allows for such a
set $\cC_n$ and for any binary tree $T$. 

The following result is independent of the result from the main theorem of \cite{bor15} mentioned above, in the sense that neither result directly implies the other. Our short proof involves a simple application of the probabilistic method, the Chernoff bound,  and a property of the random cluster model on trees established in \cite{mos}.

\begin{theorem}
\label{thm1}
For any two values $\alpha, \beta >0$ for which $\alpha+\beta>1$
 there is a value $N$ so that for all $n \geq N$ the following holds. For any unrooted binary  phylogenetic tree $T$ on a leaf set $X$ of size $n$ there is a set $\cC_n$ of  $k_n=\lfloor n^\alpha \rfloor $ characters on $X$ that capture $T$, and  for which each character takes a most  $r_n = \lfloor n^\beta\rfloor$ distinct states. 
\end{theorem}

\begin{proof}
Consider the following random process performed on $T$.  Each edge of $T$ is independently cut with probability $p_n=r_n/4n$, or is left intact with probability $1-p_n$. This leads to a partition of $X$ corresponding to the equivalence relation that two leaves are related if and only if they lie in the same connected component of the resulting graph.  We will associated to each such partition a character that induces this same partition (e.g. the character $f: X \rightarrow 2^X$  which maps $x$ to its equivalence class under $\sim$). Notice the number of `states' of this associated character is simply the number of blocks of the original partition). 

Let $Y$ denote the random number of edges of $T$ that are cut.  Then $Y$ has a binomial distribution $Y\sim {\rm Bin}(2n-3, p_n)$, which has mean $\mu_n= (2n-3)p_n = (\frac{1}{2}- o(1))n^\beta$.
By a multiplicative form of the `Chernoff bound' in probability theory ({\em c.f.} \cite{hag}, Eqn. (6) with $\epsilon =1$) we have
$\PP(Y \geq  2\mu_n) \leq \exp(-\mu_n/3)$
and since $r_n > 2\mu_n$ we obtain:
\begin{equation}
\label{eq1}
\PP(Y\geq r_n) \leq \exp(-\mu_n/3).
\end{equation}

The number of blocks of the partition of $X$ induced by randomly cutting edges of $T$ in the process described is at most $Y+1$.
Thus, the probability that a character, generated by the random cluster model with $p_n$ value as specified, takes strictly more than $r_n$ states is at most $\PP(Y+1 > r_n) = \PP(Y \geq  r_n) \leq \exp(-\mu_n/3)$, by (\ref{eq1}).

Let us generate a set $\cC_n$ of  $k_n$ such characters independently by the process described  (i.e. constructing partitions of $X$ and for each partition giving an associated character). 
The probability that at least one of these characters has more than $r_n$ states is, by Boole's inequality, at most $$n^\alpha \exp(-\mu_n/3) = n^\alpha \exp\left(-\frac{1}{3} (\frac{1}{2}-o(1))n^\beta\right) \rightarrow 0$$ as $n\rightarrow \infty$ (recall $\beta>0$). 
 Thus, there exists some value $N_1$ for which, for any $n\geq N_1$, , at least one character in $\cC_n$ takes more than $r_n$ states with probability at most  $\nicefrac{1}{3}$

What is the probability that $\cC_n$ captures $T$?   As part of a more general analysis of the (infinite state) random cluster model 
by \cite{mos},  Lemma 2.2 and Theorem 2.4 of that paper show that $\cC_n$  captures $T$ with probability at least $1 - \epsilon$ 
provided that
$k = \lceil \frac{1}{B} \log(n^2/\epsilon)\rceil,$
where $B = p_n(2-\frac{1}{1-p_n})^4 \sim p_n$ (as $n \rightarrow \infty$). 
Now, $$\frac{1}{B} \sim \frac{1}{p_n}=\frac{4n}{r_n} \sim \frac{4n}{n^\beta} = 4n^{1-\beta},$$
and since $\alpha+\beta>1$ it follows that for any $\epsilon>0$:
$$\frac{1}{B} \log(n^2/\epsilon)/n^\alpha  \sim4n^{1-\beta}\log(n^2/\epsilon)/n^\alpha \rightarrow 0 \mbox { as } n\rightarrow \infty.$$ So taking  $\epsilon = \nicefrac{1}{3}$, 
 there is a value $N_2$ for which, for any $n \geq N_2$, we have
$\lceil \frac{1}{B} \log(n^2/\epsilon)\rceil \leq \lfloor n^\alpha \rfloor$ for all $n\geq N_2$.
Thus, with $k_n = \lfloor n^\alpha\rfloor$ where $n\geq N_2$, $\cC_n$ fails to capture $T$
with probability at most $\nicefrac{1}{3}$. 

Combining these two observations, if we set $N= \max\{N_1, N_2\}$ then for all $n \geq N$, the probability that a set $\cC_n$ of $\lfloor n^\alpha \rfloor$  randomly-generated characters satisfies at least one of the following properties:
\begin{itemize}
\item[(i)]  $\cC_n$ contains a character that takes more than  $r_n$ states, or
\item[(ii)]  $\cC_n$ fails to captures $T$,
\end{itemize}
is at most  $\nicefrac{1}{3}+\nicefrac{1}{3} =\nicefrac{2}{3}$, by Boole's  inequality.
Thus there is a strictly positive probability that $\cC_n$ satisfies neither of condition (i) and (ii), and so there must exist a set of
$ \lfloor n^\alpha \rfloor$ characters, each taking at most $ \lfloor n^\beta \rfloor$ states,  which captures $T$.
This completes the proof. 
\end{proof}

{\bf Remark:} Notice from the proof, that  the condition $\alpha+\beta > 1$ can be replaced by $\alpha+\beta=1$ if we allow 
$ k_n = \lfloor n^\alpha \rfloor$ characters to be replaced by $k_n =  \lfloor n^\alpha \rfloor (8+c)\log(n)$, for any $c>0$. 


\begin{thebibliography}{99}


\bibitem{bor15} Bordewich, M. and Semple, C. (2015). Defining a phylogenetic tree with the minimum number of $r$-state characters. 
SIAM J. Discr. Math. 29(2): 835--853.
\bibitem{4char1} Bordewich, M., Semple, C. and Steel, M. (2006). Identifying X-trees with few characters. Electronic Journal of Combinatorics 13 \#R83.
\bibitem{hag} Hagerup, T. and R{\"u}b, C. (1990). A guided tour of Chernoff bounds, Information Processing Letters, 33(6), 305--308.
\bibitem{4char2}  Huber, K., Moulton, V.  and Steel, M.  (2005). Four characters suffice to convexly define a phylogenetic tree. SIAM Journal on Discrete Mathematics 18(4): 835--843.
\bibitem{mos} Mossel, E. and Steel, M. (2004).  A phase transition for a random cluster model on phylogenetic trees.
Mathematical Biosciences 187: 189--203.
\bibitem{sem} Semple, C. and Steel, M. (2002). Tree reconstruction from multi-state characters, Adv. Appl. Math. 28(2): 169--184



\end{thebibliography}
\end{document}